\def\notes{0}
\def\fsttcs{0}

\ifnum\fsttcs=1
\documentclass[numberwithinsect,a4paper,UKenglish]{lipics}
\usepackage{microtype}
\else
\documentclass[11pt]{article}
\fi

\ifnum\fsttcs=0
\usepackage[bookmarks,colorlinks,breaklinks]{hyperref}  % PDF hyperlinks, with coloured links
\hypersetup{linkcolor=blue,citecolor=blue,filecolor=blue,urlcolor=blue} % all blue links
\usepackage{fullpage}
\else
\usepackage{hyperref}
\fi
\usepackage{amsfonts}
\usepackage{amssymb}
\usepackage{amsmath}
\ifnum\fsttcs=0
\usepackage{amsthm}
\fi
\usepackage{sectsty}
\usepackage{color}
\usepackage{underscore}  %%% To avoid LaTeX complaining of _ in non-math text

\ifnum\fsttcs=0
\theoremstyle{plain}
\newtheorem{theorem}{Theorem}[section]
\newtheorem{lemma}[theorem]{Lemma}

\newtheorem{corollary}[theorem]{Corollary}

\theoremstyle{definition}
\newtheorem{definition}[theorem]{Definition}

\newtheorem{fact}[theorem]{Fact}

\fi
\theoremstyle{plain}
\newtheorem{claim}[theorem]{Claim}

\newcommand{\defeq}{\stackrel{\mathrm{def}}{=}}

\def\({\left(}
\def\){\right)}

\renewcommand{\epsilon}{\varepsilon}
\renewcommand{\phi}{\varphi}

\def\X{\mathcal{X}}
\def\Y{\mathcal{Y}}
\def\Z{\mathcal{Z}}
\def\W{\mathcal{W}}

\def\R{\mathcal{R}}

\def\S{\mathcal{S}}

\def\sR{\mathsf{R}}
\def\sD{\mathsf{D}}

\def\ve{{\varepsilon}}

\mathchardef\mhyphen="2D

\newcommand{\suppress}[1]{}
\newcommand{\etal}{\emph{et al.\/}}

\newcommand {\br} [1] {\ensuremath{ \left( #1 \right) }}
\newcommand {\Br} [1] {\ensuremath{ \left[ #1 \right] }}
\newcommand {\set} [1] {\ensuremath{ \left\lbrace #1 \right\rbrace }}
\newcommand {\minusspace} {\: \! \!}
\newcommand {\smallspace} {\: \!}
\newcommand {\fn} [2] {\ensuremath{ #1 \minusspace \br{ #2 } }}
\newcommand {\Fn} [2] {\ensuremath{ #1 \minusspace \Br{ #2 } }}

\newcommand {\fndec} [3] {\ensuremath{ #1 : \, #2 \rightarrow #3 }}
\newcommand {\mutinf} [2] {\fn{\mathrm{I}}{#1 \smallspace : \smallspace #2}}
\newcommand {\condmutinf} [3] {\mutinf{#1}{#2 \smallspace \middle\vert \smallspace #3}}
\newcommand {\prob} [1] {\Fn{\Pr}{#1}}

\newcommand {\norm} [1] {\ensuremath{ \left\| #1 \right\| }}
\newcommand {\normsub} [2] {\ensuremath{ \norm{#1}_{#2} }}
\newcommand {\onenorm} [1] {\normsub{#1}{1}}

\newcommand {\relent} [2] {\fn{\mathrm{S}}{#1 \middle\| #2}}
\newcommand {\rminent} [2] {\fn{\mathrm{S}_{\infty}}{#1 \middle\| #2}}
\DeclareMathOperator*{\bigE}{\mathbb{E}}
\newcommand {\expec} [2] {\Fn{\bigE_{\substack{#1}}}{#2}}

\newcommand {\email} [1] {E-mail: \texttt{#1}.}

\newcommand {\tribes} {\ensuremath{\mathsf{Tribes}}}

\newcommand {\rec} {{\sf rectangle~bound}}
\newcommand {\srec} {{\sf smooth-rectangle~bound}}

\newcommand {\srecs} {\ensuremath{\mathsf{srec}}}

\newcommand {\infcomp} {{\sf {information~complexity}}}
\newcommand {\alice} {{\sf {Alice}}}
\newcommand {\bob} {{\sf {Bob}}}
\newcommand {\bad} {\ensuremath{\mathsf{BAD}}}
\newcommand {\bada} {\ensuremath{\mathsf{BAD_A}}}
\newcommand {\badb} {\ensuremath{\mathsf{BAD_B}}}
\newcommand {\badaorb} {\ensuremath{\mathsf{BAD_{A\vee B}}}}
\ifnum\notes=1
\newcommand{\phnote}[1]{\noindent {\sf{\color{magenta} [Prahladh: #1]}}}
\newcommand{\rjnote}[1]{\noindent {\sf{\color{red} [Rahul: #1]}}}
\else
\newcommand{\phnote}[1]{}
\newcommand{\rjnote}[1]{}
\fi
\newcommand{\lref}[2][]{\hyperref[#2]{#1~\ref*{#2}}}
\renewcommand{\eqref}[1]{\hyperref[#1]{(\ref*{#1})}}
\numberwithin{equation}{section}

\newcommand {\Rahul} {Rahul Jain}
\newcommand {\Prahladh} {Prahladh Harsha}
\newcommand{\CQTCS} {
    Centre for Quantum Technologies and Department of Computer Science,
    National University of Singapore, Singapore.
}
\newcommand {\TIFR} {
    Tata Institute of Fundamental Research, 
Mumbai, India.
}

\newcommand {\PdfTitle} {A strong direct product theorem for the tribes  function via the smooth-rectangle bound}

\title {{\PdfTitle}\ifnum\fsttcs=1\footnote{
The work of the second author, Rahul Jain, is supported by the Singapore Ministry of Education
Tier 3 Grant and also the Core Grants of the Centre for Quantum
Technologies, Singapore.}
\fi}
\ifnum\fsttcs=1
\titlerunning{The tribes function via the smooth-rectangle bound}
\author[1]{\Prahladh}
\author[2]{\Rahul}
\affil[1]{\TIFR \texttt{prahladh@tifr.res.in}}
\affil[2]{\CQTCS \texttt{rahul@comp.nus.edu.sg}}
\authorrunning{Prahladh Harsha and Rahul Jain}
\Copyright{Prahladh Harsha and Rahul Jain}
\serieslogo{}%please provide filename (without suffix)
\volumeinfo%(easychair interface)
  {Billy Editor, Bill Editors}% editors
  {2}% number of editors: 1, 2, ....
  {Conference title on which this volume is based on}% event
  {1}% volume
  {1}% issue
  {1}% starting page number
\EventShortName{}
\DOI{10.4230/LIPIcs.xxx.yyy.p}% to be completed by the volume editor
\else
\author{
\Prahladh\thanks{\TIFR \email{prahladh@tifr.res.in}} \and \Rahul\thanks{\CQTCS \email{rahul@comp.nus.edu.sg}  
Supported by the Singapore Ministry of Education
Tier 3 Grant and also the Core Grants of the Centre for Quantum
Technologies, Singapore.}}
\date {}
\fi

\begin{document}

\ifnum\fsttcs=1
\maketitle
\begin{abstract}
\else
\begin{titlepage} 
\maketitle
\abstract{
\fi

The main result of this paper is an optimal strong direct product
result for the two-party public-coin randomized communication
complexity of the \tribes\ function. This is proved by providing an
alternate proof of the optimal lower bound of $\Omega(n)$ for the
randomised communication complexity of the \tribes\ function using the
so-called \srec, introduced by Jain and Klauck~\cite{JainK2010}. The
optimal $\Omega(n)$ lower bound for \tribes\ was originally proved by
Jayram, Kumar and Sivakumar~\cite{JayramKS2003}, using a more powerful
lower bound technique, namely the \infcomp\ bound. The \infcomp\ bound
is known to be at least as strong a lower bound method as the
\srec~\cite{KerenidisLLRX2012}. % Prior, to our result, \tribes\ was the
% only known function (to the best of our knowledge) whose communication
% complexity lower bound proof required a stronger technique than the
% \srec.
On the other hand, we are not aware of any function or relation for which
the \srec\ is (asymptotically) smaller than its public-coin randomized
communication complexity.  The optimal direct product for \tribes\ is
obtained by combining our \srec\ for tribes with the strong direct
product result of Jain and Yao~\cite{JainY2012} in terms of \srec.

\ifnum\fsttcs=1
\end{abstract}
\else
}
\thispagestyle{empty}
\end{titlepage}
\fi

\section{Introduction}

Study of lower bounds for various natural functions and relations has
been a major theme of research in communication complexity from its
advent; both for its own intrinsic value and for applications of these
bounds towards other areas of theoretical computer
science~\cite{KushilevitzNisan}. Several lower bound techniques have
been developed over the years in communication complexity such as
fooling sets, discrepancy method, rectangle bound, information
complexity bound, partition bound etc. It is interesting to understand
the relative power of these techniques and rank them against each
other. Sometimes, we  would like to understand what is the weakest
technique required to prove a particular lower bound.

An important and extensively used technique in communication
complexity is the so called \rec\ (a.k.a. the
{\sf corruption~bound}). In this technique, one argues that for
some output value $z$, and all large rectangles, a constant fraction
of inputs in the rectangle have a function value different from
$z$. This helps to lower bound the distributional communication
complexity of the function, which then translates to a lower bound on
the public-coin communication complexity via Yao's minmax
principle~\cite{Yao1983}. This technique has been successfully applied
to obtain optimal lower bounds for several problems; Razborov's lower
bound proof~\cite{Razborov1992} for the set-disjointness
function~\cite{KalyanasundaramS1992} is arguably the most well-known
application of this technique.

Another technique that has been extremely useful is the \infcomp\
bound~\cite{PonzioRV2001,ChakrabartiSWY2001}. In this method, one lower bounds the
distributional communication complexity by the amount of information
the transcript of the protocol reveals about the inputs of \alice\ and
\bob. The tools from information theory then come handy to lower bound
the information cost of the protocol. Bar-Yossef, Jayram, Kumar and
Sivakumar~\cite{BarYossefJKS2004} successfully used this
technique\footnote{The notion of \infcomp\ was formalized by Chakrabarti,
Shi, Wirth and Yao [CSWY01] in the direct sum context, however has
been used by earlier works as well for example by Ponzio,
Radhakrishnan and Venkatesh~\cite{PonzioRV2001} for showing optimal
lower bounds on the communication complexity of the pointer-chasing
problem. Chakrabarti~\etal~\cite{ChakrabartiSWY2001} defined and used, what in today's
language is called,``external information cost'' while Bar-Yossef, Jayram, Kumar and
Sivakumar~\cite{BarYossefJKS2004} defined and used ``internal
information cost'' in their proof of the disjointness lower bound.} to
give an alternate proof of the linear lower bound for the
set-disjointness function. This method has also been useful to give an
optimal linear lower bound for the \tribes\
function~\cite{JayramKS2003}.

Jain and Klauck~\cite{JainK2010}, using tools from linear programming
and semi-definite programming gave a uniform treatment to several of
the existing lower bound techniques and proposed two additional lower
bound techniques, the so-called {\sf partition~bound} and the
\srec. These bounds are stronger than almost all other known lower
bound techniques including the \rec. The {\sf partition~bound},
as the name suggests, is a linear programming formulation of the
number of partitions in a randomized protocol. The \srec, a weakening of
the {\sf partition~bound}, is a robust version of the \rec\ in
the following informal sense: \srec\ for a function $f$ under a distribution
$\mu$, is the maximum
over all functions $g$ , which are close to $f$
under the distribution $\mu$, of the \rec\ of $g$. In other words, a
function $f$ is said to have a large \srec, if it is close to
some other function $g$ (under the distribution $\mu)$ which has a
large \rec, even though $f$ itself might not have a large \rec. This suffices to lower bound the communication complexity
of $f$. These new lower
bound methods have been successfully applied, for example to obtain an
optimal lower bound for the Gap-Hamming
problem~\cite{ChakrabartiR2012}. In fact we are not aware (to the best
of our knowledge) of any function or relation for which the
{\sf partition~bound} or \srec\ is (asymptotically) smaller than
its public-coin randomized communication complexity. To determine how
tight these new lower bounds are, remains an important open question
in communication complexity.

Recently, Kerenidis~\etal~\cite{KerenidisLLRX2012} showed that  the
\infcomp\ is at least as powerful as the
{\sf relaxed-partition-bound}, which is a bound intermediate
between the {\sf partition~bound} and the \srec. The relative
strengths of the \infcomp\ and {\sf partition~bound} is not yet
well understood.

Another important theme in communication complexity has been the study
of the so called strong direct-product and (the weaker) direct-sum
conjectures; again for their own intrinsic value and also for
important applications of such results in other areas of theoretical
computer science~\cite{KarchmerRW1995}. A strong direct-product conjecture for the
public-coin communication complexity of a relation $f$ would state the
following. Let $c$ be the public-coin communication complexity of $f$
(with constant error). Suppose $k$ independent instances of $f$ are
being solved using communication less than $kc$, then the overall
success would be exponentially small in $k$. In fact, the \infcomp\
was introduced initially~\cite{ChakrabartiSWY2001} as a tool to
resolve the direct sum/product question. However, despite the
considerable progress made over the last few
years~\cite{BarakBCR2010,JainPY2012}, the direct product question has
not yet been resolved. On the other hand, we are not aware of any
function or relation for which this conjecture is false. Settling this
conjecture for all relations, again is an important open question in
communication complexity.

Recently, Jain and Yao~\cite{JainY2012} proved a direct-product result
for all relations in terms of the \srec\ (\srecs). They show that for
any relation $f$, if less than $k \cdot \log \srecs(f)$ communication
(c.f., \lref[Definition]{def:srec}) is provided for solving $k$
independent copies of $f$, then the overall success is exponentially
small in $k$. This provides a recipe to arrive at strong
direct-product results for any relation $f$: by exhibiting that $\log
\srecs(f)$ provides optimal lower bound for the public-coin
communication complexity of $f$. Jain and Yao's result implies (and in
some cases reproves) strong direct product result for many interesting
functions and relations including that for the set-disjointness
function (a strong direct-product result for set-disjointness was
first shown by Klauck~\cite{Klauck2010}, again via showing that the
\srec\ of a related function is large). \rjnote{Added the following} This also strongly motivates the search of functions for which their \srec\ is asymptotically smaller than their public-coin communication complexity. This leads us to the study of the \tribes\ function as described below.

\subsection{Our result} 

In this work we are concerned with the $\tribes: \{0,1\}^n \times
\{0,1\}^n \to \{0,1\}$ function, defined as follows. $$\tribes(x,y)
\defeq \bigwedge_{i=1}^{\sqrt{n}} \left(\bigvee_{j=1}^{\sqrt{n}}\left(
    x_{(i-1) \sqrt{n} +j} \wedge y_{(i-1)\sqrt{n} +j}\right)\right).$$
As mentioned earlier, an optimal linear lower bound for \tribes\ was
shown by Jayram, Kumar and Sivakumar~\cite{JayramKS2003} using the
\infcomp\ technique.  It is to be noted that the \rec\ proves only a
$\Theta(\sqrt{n})$ lower bound and thus fails to provide an optimal
lower bound for \tribes. In fact, the primary motivation for
Jayram~\etal~\cite{JayramKS2003} to study the \tribes\ function was
the fact that it provided the first example where \infcomp\ techniques
were provably stronger than the then known ``combinatorial'' lower
bound techniques. \rjnote{Added this} Therefore it is natural to ask if \tribes\ also provides a separation between \srec\ and public-coin communication complexity, in the process also implying separation between  
%Given the recent surge in combinatorial lower bound
%techniques, one can ask if \infcomp\ techniques continue to remain
%provably stronger than the combinatorial techniques (such as
%\srec, {\sf partition-bound} etc) and in particular, if \tribes\ also
%separates 
\infcomp\ bound and \srec.
%Thus, it is natural to ask if one requires as
%powerful a technique as \infcomp\ to obtain a lower bound for \tribes\
%or will a slightly stronger bound than the \rec, such as the \srec,
%suffice. 
We consider this question in this work and answer it in the negative.
\begin{theorem}[\srec\ for \tribes]\label{thm:main}\ 

For sufficiently small $\epsilon \in (0,1)$,
$\sR^{\text{\rm pub}}_\epsilon(\tribes) \geq \log \srecs_\epsilon(\tribes) \geq
\Omega(n)$. 
\end{theorem}
Here, $\sR^{\text{\rm pub}}_\epsilon(f)$ refers to the $\epsilon$-error
public-coin randomized communication complexity of $f$.

Another important motivation for our work (besides answering the above question) is its consequence to strong direct
product. As indicated in the recipe outlined above, combining our
\srec\ for \tribes\ with the result of Jain and
Yao~\cite{JainY2012}, we obtain the following.
\begin{corollary}[strong direct product for \tribes]\label{cor:sdp-tribes}
$\sR^{\text{\rm pub}}_{1-2^{-\Omega(k)}}\left(\tribes^{(k)}\right) = \Omega(kn)$.
\end{corollary}
Here, $f^{(k)}$ refers to the $k$-wise direct product of the function
$f$. Our result (\lref[Theorem]{thm:main}) also exhibits for the first
time, an asymptotic separation between the \srec\ and the \rec\ for a
total function (previously a quadratic separation was known however for
the Gap-Hamming partial function~\cite{ChakrabartiR2012}).

% Prior to our result, \tribes\ was the only known function (to the best
% of our knowledge) whose communication complexity bound proof required
% a stronger technique than the \srec.
It is to be noted that the
\infcomp\ lower bound for \tribes\ was generalised to constant depth
read-once trees functions~\cite{JayramKR2009, LeonardosS2010}. Given
our results, it is interesting to ask if these lower bounds can be
obtained using the \srec\ instead, which would imply a direct product
for these functions. These alternate lower bounds might also help to
obtain bounds for super-constant depth read-once formulae.

\subsection{Our techniques}\label{sec:techniques}

It will be convenient for us to view the \tribes\ function as the
conjunction of $\sqrt{n}$ set-disjointness functions over $\sqrt{n}$
sized inputs\footnote{\phnote{15 Jul: clarified
    disjointess/intersection confusion}By the disjointness function, we refer to the
  function $\bigvee_{j=1}^{\sqrt{n}}\left(
    x_{j} \wedge y_{j}\right)$. Strictly speaking, this is the
  set-intersection problem, but as is common in this literature, we
  will abuse notation and refer to this problem as the
  set-disjointness problem.}. We refer to the $\sqrt{n}$ sized inputs to each of the
disjointness functions as a block. We consider a distribution $\mu$ on
the inputs for the \tribes\ function which has support only on the
following type of inputs: in every block, except for one block (say
$j$), the inputs to the two parties \alice\ and \bob\ are NO instances of
the disjointness function (the sets corresponding to the blocks
intersect at exactly one location) and in block $j$, there could be $0,1$ or $2$ intersections
which occur at locations $k_j$ and $l_j$. Let's refer to the three
types of subsets of inputs based on the number of intersections as
$U_0, U_1$, and $U_2$ respectively. Recall that to show that the \srec\ of
\tribes\ is large, we need to demonstrate a function $g$, close to
\tribes\ (under $\mu$), whose \rec\ is large. This function $g$ is
constructed as follows: $g$ takes  value $0$ in $U_0 \cup U_2$ and
value $1$ in $U_1$. Note that \tribes\ takes value $0$ in $U_0$ and
value $1$ in $U_1 \cup U_2$. I.e., \tribes\ and $g$ disagree on the
inputs in $U_2$. For our choice of distribution $\mu$, this
disagreement set $U_2$ will have weight $\mu(U_2) \approx 1/16$ while the weight
of the $1$-inputs will be approximately $\mu(U_1) \approx 6/16$ (i.e., $U_1$
is 6 times larger than $U_2$).

Observe that for \tribes, there are large rectangles (of size $\approx
2^{-\sqrt{n}}$ under $\mu$) which are monochromatic. We can just fix
any one coordinate in each block and force intersection there to
create large $1$-monochromatic rectangle. Similarly we can choose any
one block and force non-intersection in that entire block to create
large $0$-monochromatic rectangle. Hence the \rec\ of $\tribes$ is at
most $O(\sqrt{n})$. However, note that the $1$-monochromatic
rectangles described above are not monochromatic in $g$. Indeed, we
show that there exists constants $C $ and $D$ such that for
every large rectangle $W$ (with $\mu(W) \geq 2^{-\Omega(n)}$), 
$\mu(U_1\cap W)$ is either dominated by $C\cdot \mu(U_0\cap W)$
(this is similar to the \rec) or is dominated by $D \cdot
\mu(U_2\cap W)$. This immediately implies the \rec\ of $g$ is
$\Omega(n)$. We will prove the above statement for $D$ strictly
smaller than $6$. This fact implies that whenever $\mu(U_1\cap W)$ is not
dominated by $C\cdot \mu(U_0\cap W)$ in $W$, the ratio of $U_2$-inputs
to $U_1$-inputs in the rectangle $W$ is considerably more than the
similar ratio globally (which is $\approx 1/6$). This fact lets us 
translate the $\Omega(n)$ \rec\ for $g$ to a similar \srec\ for
\tribes.

We consider an
exhaustive collection of sub-events such that conditioned on any such
sub-event, the non-product distribution $\mu$ becomes a product
distribution. Such handling of non-product distributions, by
decomposing them into several product distributions, has been done
several times before, for instance in Razborov's
proof~\cite{Razborov1992} of the optimal lower bound for the
set-disjointness function. Assume such a conditioning exists for the
rest of this proof outline.

How does one prove that for all large rectangles $W$, either
$\mu(U_1\cap W) \leq C \mu(U_0\cap W))$ or $\mu(U_1 \cap W) \leq
D \mu(U_2 \cap W)$ for some $D$ strictly smaller than 6. Note
that one cannot prove for all rectangles $W$, $\mu(U_1 \cap W) \leq
D \mu(U_2 \cap W)$ for some $D$ strictly less than 6, since
this is false globally (i.e., $\mu(U_1) \approx 6 \mu(U_2)$). Hence,
one needs to do a case analysis\footnote{Such a case analysis is not
  required to prove \rec\ (c.f., proof of
  disjointness~\cite{Razborov1992}), but is necessary while proving a
  \srec.}. And we do this based on the values of
$\prob{X_{l_j}=Y_{l_j} =1}$ and $\prob{X_{k_j}=Y_{k_j} = 1}$.

Consider the case when $\prob{X_{l_j}= Y_{l_j}=1 } \geq \frac34
\mu(U_1 \cap W)$. Since the rectangle is large, using an entropy
argument, we can argue that in most cases, conditioned on the
sub-event $(X_{l_j}= Y_{l_j}=1)$, both $\prob{X_{k_j}=1}$ and
$\prob{Y_{k_j}=1}$ are large enough $(\approx 1/2$). Now since the distribution is
product it means that conditioned on $(X_{l_j}= Y_{l_j}=1)$,
$\prob{X_{k_j} = Y_{k_j}=1}$ is large enough and hence $\mu(U_2 \cap
W)$ is a required fraction of $\mu(U_1 \cap W)$. Similar arguments
hold for the case with the roles of $l$ and $k$ reversed.

In the third case, when $\max\{\prob{X_{l_j}= Y_{l_j}=1 },
\prob{X_{k_j}= Y_{k_j}=1 } \} \leq \frac34 \mu(U_1 \cap W)$, again
using the same entropy argument, we can show that $\prob{X_{l_j}=
  Y_{l_j}=1, X_{k_j}= Y_{k_j}=0 }$ and $\prob{X_{l_j}=Y_{l_j}=0,
  X_{k_j}= Y_{k_j}=1 }$ are large. Now, since $W$ is a rectangle, we
can show that $\prob{X_{l_j}=1, Y_{l_j}=0, X_{k_j}= 0, Y_{k_j}=1 }$
and $\prob{X_{l_j}=0, Y_{l_j}=1, X_{k_j}= 1, Y_{k_j}=0 }$ are large
using a {\em cut-and-paste} argument. This implies that $\mu(U_0 \cap
W)$ is a required fraction of $\mu(U_1 \cap W)$. This concludes our
proof outline.

We note that our distribution is similar to (and in fact inspired
from) the distribution used by Jain and
Klauck~\cite{JainK2010} while analyzing the query complexity of the
\tribes\ function. We also note  that the distribution used by Jayram,
Kumar and Sivakumar~\cite{JayramKS2003} in their \infcomp\ lower bound
for \tribes\ is different from our distribution, in particular,
their distribution does not put any support on $U_2$ inputs which
have intersections of size 2 within block $j$. However, we do add that
they also use similar in spirit, albeit different cut-and-paste
arguments in their lower bound proof.

%\subsubsection*{Organization} 

\section{Preliminaries}
\label{sec:Preliminaries}

\paragraph*{Communication Complexity:}
We begin by recalling the Yao's two-party communication
model~\cite{Yao1979} (see Kushilevitz and
Nisan~\cite{KushilevitzNisan} for an excellent introduction to the
area). Let $\X$, $\Y$ and $\Z$ be finite non-empty sets, and let $f:
\X \times \Y \to \Z$ be a function. A two-party protocol for computing
$f$ consists of two parties, {\alice} and {\bob}, who get inputs $x
\in \X$ and $y \in \Y$ respectively, and exchange messages in order to
compute $f(x,y) \in \Z$.

For a distribution $\mu$ on $\X \times \Y$, let the $\epsilon$-error
distributional communication complexity of $f$ under $\mu$ (denoted by
$\sD_{\epsilon}^{\mu}(f)$), be the number of bits communicated (for the
worst-case input) by the best deterministic protocol for $f$ with
average error at most $\epsilon$ under $\mu$.  Let $\sR^{{\text{\rm pub}}}_{\epsilon}(f)$, the public-coin randomized communication
complexity of $f$ with worst case error $\epsilon$, be the number of
bits communicated (for the worst-case input) by the best public-coin
randomized protocol, that for each input $(x,y)$ computes $f(x,y)$
correctly with probability at least $1-\epsilon$. Randomized and
distributional complexity are related by the following celebrated
result of Yao~\cite{Yao1983}.
\begin{theorem}[Yao's minmax principle~\cite{Yao1983}]
\label{thm:yao}
$\sR^{\text{\rm pub}}_{\epsilon}(f) = \max_{\mu} \sD_{\epsilon}^{\mu}(f) $.
\end{theorem}

Given a function $f: \X \times \Y \to \Z$, the $k$-wise
direct product of $f$, denoted by $f^{(k)}$ is the function $f:\X^k \times
\Y^k \to \Z^k$ defined as follows:
$f^{(k)}((x_1,\dots,x_k),(y_1,\dots,y_k)) =
(f(x_1,y_1),\dots,f(x_k,y_k))$. The direct product/sum question
involves relating $\sR^{\text{\rm pub}}(f^{(k)})$ to $\sR^\text{\rm pub}(f)$. More
precisely, the strong direct product conjecture states that
$\sR^{\text{\rm pub}}_{1-2^{-\Omega(k)}}(f^{(k)}) =
\Omega\left(k \cdot \sR^{\text{\rm pub}}_{1/3}(f)\right)$.

\paragraph*{The smooth rectangle bound:}

The smooth rectangle bound was introduced by Jain and
Klauck~\cite{JainK2010}, as a generalization of the rectangle
bound. Informally, the \srec\ for a function $f$ under a distribution
$\mu$, is the maximum over all functions $g$ , which are close to $f$
under the distribution $\mu$, of the \rec\ of $g$.  However, it will
be more convenient for us to work with the following linear
programming formulation of \srec. Please see \cite[Lemma~2]{JainK2010}
and \cite[Lemma~6]{JainY2012} for the relations between the LP
formulation and the more ``natural'' formulation in terms of \rec. A
broad connection between the two definitions is that the variable
$\phi$ in the dual of the linear programming definition takes non-zero
values precisely at the inputs $(x,y)$ where $f$ and $g$ differ.

\begin{definition}[smooth-rectangle bound] \label{def:srec}
For a total Boolean function $f$, the $\epsilon$- smooth rectangle bound
of $f$ denoted $\srecs_\epsilon(f)$ is defined to be
$\max\{\srecs^z_\epsilon(f): z\in \{0,1\}\}$, where
$\srecs^z_\epsilon(f)$ is given by the optimal value of the following
linear program (below $\W$ represents the set of all rectangles in $\X \times \Y$). 

{\scriptsize
% \hspace{-0.4in}
% \begin{minipage}{2in}\vspace{0.1in}
%     \centerline{\underline{Primal}}\vspace{-0.1in}
%     \begin{align*}
%       & \text{min:}\quad  \sum_{W \in \W}  v_{W} \\
%        \quad &  \forall (x,y) \in  f^{-1}(z): \sum_{W: (x,y) \in W} v_{W} \geq 1 - \epsilon,\\
%       & \forall (x,y) \in  f^{-1}(z): \sum_{W: (x,y) \in W} v_{W} \leq 1,\\
%       &  \forall (x,y) \notin   f^{-1}(z): \sum_{W: (x,y) \in W} v_{W} \leq \epsilon,\\
%       & \forall W : v_{W} \geq 0 \enspace .
%     \end{align*}
% \end{minipage}
% \begin{minipage}{3in}\vspace{-0.6in}
%     \centerline{\underline{Dual}}\vspace{-0.1in}
%     \begin{align*}
%       & \text{max:}\quad   \sum_{(x,y)\in f^{-1}(z)} \left((1-\epsilon) \lambda_{x,y} - \phi_{x,y} \right)- \sum_{(x,y)\notin  f^{-1}(z)} \epsilon \cdot \lambda_{x,y}\\
%        \quad &   \forall W : \sum_{(x,y)\in f^{-1}(z)\cap W} (\lambda_{x,y} - \phi_{x,y}) - \sum_{(x,y)\in (W  \setminus f^{-1}(z))} \lambda_{x,y} \leq 1,\\
%       & \forall (x,y) : \lambda_{x,y} \geq 0 ; \phi_{x,y} \geq 0 \enspace .
%     \end{align*}
% \end{minipage}
\ifnum\fsttcs=1
\hspace{-0.4in}\begin{minipage}{2.3in}\vspace{0.1in}
    \centerline{\underline{Primal}}\vspace{-0.1in}
    \begin{align*}
      \text{min} \sum_{W \in \W}  v_{W} \\
       \sum_{W \ni (x,y)} v_{W} &\geq 1 - \epsilon,&& \forall (x,y) \in  f^{-1}(z)\\
      \sum_{W \ni (x,y)} v_{W} &\leq 1,&& \forall (x,y) \in  f^{-1}(z)\\
      \sum_{W\ni (x,y)} v_{W} &\leq \epsilon,&& \forall (x,y) \notin   f^{-1}(z)\\
       v_{W} &\geq 0, && \forall W\ .
    \end{align*}
\end{minipage}
\begin{minipage}{3.2in}\vspace{-0.6in}
    \centerline{\underline{Dual}}\vspace{-0.1in}
    \begin{align*}
      \text{max}\sum_{(x,y)\in f^{-1}(z)} \left((1-\epsilon)
        \lambda_{x,y} - \phi_{x,y} \right)- \sum_{(x,y)\notin
        f^{-1}(z)} \epsilon \cdot \lambda_{x,y}\\ 
       \sum_{(x,y)\in W_z}(\lambda_{x,y} - \phi_{x,y}) -
       \sum_{(x,y)\in W_{-z}} \lambda_{x,y} \leq 1,
       \quad \forall W\\
      \lambda_{x,y}, \phi_{x,y} \geq 0, \quad \forall (x,y)\\\
      \text{where } W_z = W \cap f^{-1}(z) \text{ and } W_{-z} =
      W\setminus f^{-1}(z) \enspace.
    \end{align*}
\end{minipage}
\else
\begin{align*}
\underline{Primal} &&&&\underline{Dual}\qquad\qquad\qquad\qquad\qquad\qquad\\
\text{min} \sum_{W \in \W}  v_{W} &&&&  \text{max}\sum_{(x,y)\in f^{-1}(z)} \left((1-\epsilon)
        \lambda_{x,y} - \phi_{x,y} \right)- \sum_{(x,y)\notin
        f^{-1}(z)} \epsilon \cdot \lambda_{x,y}\\ 
       \sum_{W \ni (x,y)} v_{W} &\geq 1 - \epsilon,&& \forall\  (x,y)
       \in  f^{-1}(z) &      \sum_{(x,y)\in W \cap f^{-1}(z)}(\lambda_{x,y} - \phi_{x,y}) -
       \sum_{(x,y)\in W\setminus f^{-1}(z)} \lambda_{x,y} \leq 1,
       & \qquad\forall\  W \in \W\\
    \sum_{W \ni (x,y)} v_{W} &\leq 1,&& \forall\  (x,y) \in  f^{-1}(z)
    &\lambda_{x,y} \geq 0, &\qquad\forall\  (x,y) \\
     \sum_{W\ni (x,y)} v_{W} &\leq \epsilon,&&\forall\  (x,y) \notin
     f^{-1}(z)
     & \phi_{x,y} \geq 0, &\qquad \forall\  (x,y) \enspace .\\
     v_W &\geq  0, && \forall\  W \in \W \enspace .
\end{align*}
\fi
}
\end{definition}

\begin{theorem}[{\cite[Theorem~1]{JainK2010}}]For all functions $f:\X
  \times \Y \to \{0,1\}$ and $\epsilon \in (0,1)$, we have $\sR^{\text{\rm pub}}_\epsilon(f) \geq \log (\srecs_\epsilon(f))$.
\end{theorem}

Jain and Yao~\cite{JainY2012} proved the following strong direct
product theorem in terms of the smooth rectangle bound.

\begin{theorem}[{\cite[Theorem~1 and Lemma~6]{JainY2012}}]\label{thm:jy-sdp} Let $f: \X
  \times \Y \to \{0,1\}$ be a Boolean function. For every $\epsilon
  \in (0,1)$, there exists small enough $\eta \in (0,1/3)$ such that
  the following holds. For all integers $k$,
$$\sR^{\text{\rm pub}}_{1-(1-\eta)^{\lfloor \eta^2 k/32\rfloor}} (f^{(k)}) \geq \frac{\eta^2}{32}\cdot k \cdot
  \left(11 \eta \cdot \log
    \srecs_{\epsilon} (f) - 3 \log \frac1\epsilon
    -2\right).$$
\end{theorem}

\paragraph*{Information theory:}  We need the following basic facts
from information theory. Let $\mu$ be a (probability)
distribution on a finite set $\X$ and $X$ be a random variable
distributed according to $\mu$.
Let $\mu(x)$ represent the probability of $x\in\X$ according to
$\mu$. The entropy of $X$ is defined as
$H(X)\defeq \sum_x\mu(x) \cdot \log\frac{1}{\mu(x)}.$ Entropy satisfies subadditivity:
$H(XY) \leq H(X) +H(Y).$

\section{The smooth rectangle bound  for Tribes}
\label{sec:srectribes}

% The broad outline of the \srec for \tribes\ proved in this section is similar to the
% Razborov's \rec bound for \disj, as presented in
% Kushilevitz-Nisan~\cite{KushilevitzN}. 

% The $\tribes: \{0,1\}^n \times \{0,1\}^n \to \{0,1\}$ function is defined as follows.  Let $x,y \in  \{0,1\}^n$.
% % and let $z \in \{0,1\}^n$ be the bit-wise \myand of $x$ and $y$.
% % Let $x^i,y^i,z^i$ represent the $i$th block of $\sqrt{n}$ bits of
% % $x,y,z$ respectively (there are $\sqrt{n}$ blocks of size
% % $\sqrt{n}$ in each case).
% Then $\tribes(x,y) \defeq \myand_{i=0}^{\sqrt{n}-1} (\myor_{j=1}^{\sqrt{n}}( \myand(x_{i \sqrt{n} +j}, y_{i\sqrt{n} +j})))$.

In this section, we prove a linear lower bound on the randomized
communication of \tribes\  via the \srec. 

First we introduce some notation. We will prove the result for $n$ of the form
$(2r+1)^2$, where $r\geq 2$ is even. Assume the input indices $[{n}]$ to the
\tribes\ function are partitioned into $\sqrt{n}$ blocks $s_1,\dots,
s_{\sqrt{n}}$, where the $i^{th}$ block $s_i =
\{(i-1)\sqrt{n}+1, \ldots, i \sqrt{n}\}$. Thus, $$\tribes(x,y) =
\bigwedge_{i=1}^{\sqrt{n}}\left(\bigvee_{j \in s_i} ( x_j \wedge
  y_j)\right).$$  A string $x \in \{0,1\}^n$ can be
viewed both as an $n$-bit string and as a subset $x \subseteq [n]$. We
will use both these interpretations. 

\phnote{15 Jul: Added informal description of $\mu$}
Consider the distribution $\mu(x,y)$ on the inputs of the \tribes\
function defined by the following (informal) description. As mentioned
earlier, this distribution is inspired by the distribution used by
Jain and Klauck~\cite{JainK2010} while analyzing the query complexity of
the \tribes\ function. Among the $\sqrt{n}$ blocks, one of the blocks
is chosen as a special block, say block $j$. \alice's and \bob's
inputs are then chosen such that their inputs when restricted to any of
the blocks (special or non-special) have exactly $(r/2+1)$ ones each.
Furthermore, for each of the non special blocks, \alice's and \bob's
input are chosen such that their inputs, restricted to this block,
have a unique intersection (this is identical to the yes instances of
Razborov's distribution for disjointness) while for the special block
$j$, \alice's and \bob's inputs are chosen such that their inputs,
restricted to the special block, have an intersection of size 0, 1 or
2. As in the case of Razborov's distribution, the variable $t$ is used
to denote the random variable containing the index of the special
block $j$ and other relevant information such that conditioned on $t$,
the distribution $(X,Y)$ is a product distribution. The formal
description of the distribution $\mu$ is as follows:

\phnote{15 Jul: Added missing $d_j$ in special block $j$}
\begin{enumerate}
\item Choose $j \in [\sqrt{n}]$ uniformly.\\For each $i \in
  [\sqrt{n}]\setminus \{j\}$, randomly partition the indices in
  $s_i$ as follows: $s_i=(t^A_i,
  t^B_i, \{l_i\})$ into 3 disjoint sets such that $|t^A_i| = |t^B_i| =
  r$ and $l_i \in s_i$.\\ For index $j$, randomly partition
  the indices in $s_j$ as follows: $s_j = (\tilde{t}^A_j,\tilde{t}^B_j,\{k_j\},\{l_j\},\{d_j\})$ into 5
  disjoint sets such that $|\tilde{t}^A_j| = |\tilde{t}^B_j| = r-1$
  and $k_j,l_j,d_j\in s_j$. Set $t^A_j = \tilde{t}^A_j \cup \{k_j\}$
  and $t^B_j = \tilde{t}^B_j \cup \{k_j\}$. \\
Let $t = \left(j, k_j, (t^A_i, t^B_i , l_i)_{i \in
      [\sqrt{n}]}\right)$.

\item For each $i \neq j \in [\sqrt{n}]$, set the variables in block
  $s_i$ as follows:
  \begin{itemize}
  \item Set $x_{l_i} \gets 1$ and $x_{s_i \setminus (t^A_i\cup
      \{l_i\})}\gets \bar{0}$. Let $x_{t^A_i}$ be a random string of
    exactly $r/2$ ones.
    \item Set $y_{l_i} \gets 1$ and $y_{s_i\setminus (t^B_i\cup
        \{l_i\})}\gets \bar{0}$.  Let $y_{t^B_i}$ be a random string
      of exactly $r/2$ ones.
  \end{itemize}
\item Set the variables in block $s_j$ as follows:
  \begin{itemize}
  \item Let $x_{t^A_j \cup \{l_j\}}$ be a random string of exactly $r/2+1$ ones and $x_{s_j \setminus (t^A_j\cup       \{l_j\})}\gets \bar{0}$.
  \item Let $y_{t^B_j \cup \{l_j\}}$ be a random string of exactly $r/2+1$ ones and $y_{s_j \setminus (t^B_j\cup \{l_j\})}\gets \bar{0}$.
  \end{itemize}
\end{enumerate}

Let $(X,Y)$ be distributed according to $\mu$, where $X$ represents
the input to \alice\ and $Y$ represents the input to \bob. Let $T=
\left(J, K_J, (T^A_i, T^B_i , L_i )_{i \in [\sqrt{n}]}\right)$ be the
random variable (correlated with $(X,Y)$) representing $t$ distributed
as above.  Observe that though $(X,Y)$ is not a product distribution,
the conditional distribution $((X,Y)~|~T=t)$ is product for each $t$.

Partition the set of inputs (in the support of $\mu$) into 3 sets
$U_0, U_1$ and $U_2$ as follows:
$$ U_i =  \{ (x,y) ~|~ \mu(x,y) > 0 \text{ and sets $x$ and $y$ have
  exactly $\sqrt{n}-1 +i$ intersections}\}.$$
%$$U_i = \{(x,y) ~|~ \mu(x,y) > 0 \text { and blocks $x_{s_j}$ and $y_{s_j}$ have exactly $i$ intersections} \}.$$ 
Note that $U_0$ are
the $0$-inputs and $U_1 \cup U_2$ the $1$-inputs of the \tribes\ 
function while $U_0 \cup U_2$ and $U_1$ are the $0$- and $1$-inputs
respectively of the function $g$ described in 
\lref[Section]{sec:techniques}.  

Let $\beta \defeq  \frac{r+2}{r+1}$. The following facts can be easily
verified from the definition of the distribution $\mu$. For all $t$,
\begin{equation*}
\prob{X_{l_j} = 1 ~|~ T=t} = \frac\beta2; \quad \prob{X_{l_j} = X_{k_j}=1 ~|~ T=t} = \prob{X_{l_j} = 1,
  X_{k_j}=0 ~|~ T=t} = \frac\beta4.
\end{equation*}
Given this, it can be easily checked that the weights of the sets
$U_0, U_1$ and $U_2$ are as follows: $\mu(U_0) = 1 - 7\beta^2/16,
\mu(U_1) = 6\beta^2/16$, and $\mu(U_2) = \beta^2/16$. 

Our main lemma is the following (we have not optimized the
constants).
\begin{lemma}\label{lem:oneisdominated} There exists a constant
  $\delta \in (0,1)$ such that for sufficiently large $n$, the
  following holds: for every rectangle $W = A \times B$,
  we have
$$0.99 \mu(U_1 \cap W) \leq \frac{16}{3(0.99)^2}\cdot \mu (U_2
  \cap W) + \frac{16}{(0.99)^2} \mu(U_0 \cap W) + 2^{-\delta n/2 +1}.$$
\end{lemma}
In other words, in any rectangle which contains a significant fraction
of inputs from $U_1$ (i.e., at least $2^{-\delta n/2+1}$), the weight of the $U_1$
inputs is dominated by some linear function of the weights of $U_0$ and
$U_2$ inputs. Before proving this lemma, let us first see how this
lemma implies the \srec\  for \tribes, which implies our \lref[Main Theorem]{thm:main}
\begin{theorem}[\srec\  for \tribes] There exists $\gamma \in (0,1)$
  such that for all sufficiently large $n$ and $\epsilon < 1/1000$, we
  have: 
$\srecs^{1}_\ve(\tribes) \geq 2^{\gamma \cdot n}$.
\end{theorem}
\begin{proof} 
We will prove the bound using the dual formulation for \srec\ given in \lref[Definition]{def:srec}.
Define the dual variables $\lambda_{x,y}$ and $\phi_{x,y}$ as follows:
\begin{align*}
\lambda_{x,y}& = 
\begin{cases}
   0  & \text{ if } \quad (x,y) \in U_2\\
  0.99 \mu(x,y) 2^{\delta n/2-1} & \text{ if } \quad (x,y) \in U_1\\
 \frac{16}{(0.99)^2}  \mu(x,y) 2^{\delta n/2-1}  & \text{ if } \quad
 (x,y) \in U_0.\\
\end{cases}\\
\phi_{x,y}  &= 
\begin{cases}
\frac{16}{3(0.99)^2} \mu(x,y) 2^{\delta n/2-1}  & \text{ if } \quad  (x,y) \in U_2\\
0  & \text{ if } \quad (x,y) \in U_1 \cup U_0.
\end{cases}
\end{align*}
From \lref[Lemma]{lem:oneisdominated} we get 
 $$\forall \text{ rectangles } W : \sum_{(x,y)\in \tribes^{-1}(1)\cap W} (\lambda_{x,y} - \phi_{x,y}) - \sum_{(x,y)\in (W  \setminus \tribes^{-1}(1))} \lambda_{x,y} \leq 1 .$$
The objective of the LP can be bounded as follows:
\begin{align*}
& \sum_{(x,y)\in \tribes^{-1}(1)} \left((1-\epsilon) \lambda_{x,y} -
  \phi_{x,y} \right)- \sum_{(x,y)\notin  \tribes^{-1}(1)} \epsilon
\cdot \lambda_{x,y}\\
& \geq \left( (0.999)(0.99) \mu(U_1) - \frac{16}{3(0.99)^2}\mu(U_2) -
  \frac{16}{1000(0.99)^2}\mu(U_0)\right)2^{\delta n/2-1}\\
& \geq 0.02\cdot 2^{\delta n /2 -1} \qquad \text{(for sufficiently
  large n)}.
\end{align*}
Thus, proved.
\end{proof}

\lref[Corollary]{cor:sdp-tribes} follows by combining the above
theorem and Jain-Yao's strong direct product theorem in terms of the
\srec\ (\lref[Theorem]{thm:jy-sdp}).

\subsection{Proof of Lemma~\ref*{lem:oneisdominated}}

Let $W=A\times B$ be the rectangle.  For each $t = \left(j,k_j,
  (t^A_i,t^B_i,l_i)_{i \in [\sqrt{n}]}\right)$ and $a,b\in
\{0,1\}$, define,
\begin{align*}
R(t, a, b) = \prob{X \in A ~|~ T=t, X_{l_j} = a, X_{k_j} = b}&, R(t, a) = \prob{X \in A ~|~ T=t, X_{l_j} = a},\\
C(t, a, b) = \prob{Y \in B ~|~ T=t,  Y_{l_j} = a, Y_{k_j} =b}&, C(t, a) = \prob{Y \in B ~|~ T=t,  Y_{l_j} = a}.
\end{align*}

Define the following random variables (we will set $\delta$ later): $$\bada(t) = 1 \text{
  iff }\min\{R(t,1,1),R(t,1,0)\} < 0.99\left(R(t,1) - 2^{-\delta
    n}\right),$$ and symmetrically, $$\badb(t)=1 \text{ iff
}\min\{C(t,1,1), C(t,1,0)\} < 0.99\left(C(t,1) - 2^{-\delta
    n}\right).$$

For a given $t$, let $t'$ denote a partition identical to $t$ except
that the role of the indices $l_j$ and $k_j$ are exchanged (i.e.,
$k'_j = l_j, l'_j = k_j, (t^A_j)' = \tilde{t}^A_j \cup \{l_j\}$ and
$(t^B_j)'= \tilde{t}^B_j \cup \{l_j\}$). To define $\bad(t)$, we need the
following two quantities.
\begin{eqnarray*}
\rho_l(t) & = & \prob{X_{l_j} =Y_{l_j}=1, X \in A, Y \in B, (X,Y) \in U_1 ~|~ T=t},\\
\rho_k(t) & = &\prob{X_{k_j} =Y_{k_j}=1, X \in A, Y \in B, (X,Y) \in U_1 ~|~ T=t}. 
\end{eqnarray*}
Observe that $\mu(U_1 \cap W ~|~ T=t) =
\rho_l(t) +\rho_k(t)$. Hence, it must be the case that exactly one of
the following happens: (1) $\rho_l(t) > 3\mu(U_1 \cap W ~|~ T=t)/4$,
(2) $\rho_k(t) > 3\mu(U_1 \cap W ~|~ T=t)/4$ or (3) $\max\{\rho_l(t),\rho_l(t)\} \leq 3\mu(U_1 \cap W ~|~
T=t)/4$ (equivalently, $\min\{\rho_l(t),\rho_l(t)\} \geq \mu(U_1 \cap W ~|~
T=t)/4$). We define $\bad(t)$ based on these cases as follows.
\begin{eqnarray}
\bad(t) = \begin{cases} 
  \bada(t) \vee \badb(t), & \text{if }\rho_l(t) > \frac{3\mu(U_1 \cap W ~|~ T=t)}4\\
  \bada(t') \vee \badb(t'), &\text{if } \rho_k(t) > \frac{3\mu(U_1 \cap W ~|~ T=t)}4\\
  \bada(t) \vee \badb(t) \vee \bada(t') \vee \badb(t'), &
  \text{otherwise.}
\end{cases}\label{eq:badt}
\end{eqnarray}

The following claim shows that the probability that $\bada(T)$ and $\badb(T)$ occurs
is small.

\begin{claim} \label{claim:badsmall} There exists  a small fixed
  constant $\delta > 0$ such that for sufficiently large $n$, the
  following holds: for any $(t^A_i , l_i)_{i \in
    [\sqrt{n}]}$, we have 
$$\prob{\bada(T) = 1 ~|~ T^A_i = t^A_i , L_i = l_i, \text{ for each } i \in [\sqrt{n}]} < \frac{1}{6400}. $$
(Symmetrically, for any $(t^B_i,l_i)_{i}$, $\prob{\badb(T)=1
  ~|~ T^B_i = t^B_i , L_i = l_i, \mbox{ for each } i \in [\sqrt{n}]} <
\frac1{6400}$.)  
% Let $\alpha, \alpha', \delta , \beta, p > 0$ be such that $H(p (1 - \alpha)) <H(p) (1 - \alpha')$ and $\delta < \frac{1}{3}H(p) \beta \alpha'$. 
%  Let the binary random variable $\bada(t) = 1$ if  $\min\{R(t,1,1),R(t,1,0)\} < (1- \alpha)R(t,1) - 2^{-\delta n} $, and $0$ otherwise.  Similarly let the binary random variable $\badb(t)= 1$ if $\min\{C(t,1,1), C(t,1,0)\}  < (1- \alpha)C(t,1) - 2^{-\delta n} $, and $0$ otherwise.   For each $\{t^B_i , l_i  : ~ i \in [\sqrt{n}]\}$ and $\{t^A_i , l_i  : ~ i \in [\sqrt{n}]\}$  we have,  
% $$\prob{\bada(T)=1 ~|~ T^B_i = t^B_i , L_i = l_i, \mbox{ for each } i \in [\sqrt{n}]} < \beta. $$
% $$\prob{\badb(T)=1 ~|~ T^A_i = t^A_i , L_i = l_i, \mbox{ for each } i \in [\sqrt{n}]} < \beta. $$
\end{claim}
\begin{proof}
  We prove the inequality involving $\bada(T)$. The other inequality
  is proved similarly. We first consider the easy case when $(t^A_i ,
  l_i)_{i \in [\sqrt{n}]}$ satisfies $$\prob{X \in A ~|~ X_{l_i} =1,
    T^A_i = t^A_i , L_i = l_i, \mbox{ for each } i \in [\sqrt{n}]} <
  2^{-\delta n}.$$ It follows from the definition of the distribution
  $\mu$, that the above probability is unchanged on further
  conditioning by $T=t$ for any $t$ consistent with $(t^A_i ,
  l_i)_{i \in [\sqrt{n}]}$. In other words, this probability is equal to
  $R(t,1)=\prob{X \in A ~|~ T=t, X_{l_j} = 1}$ for any $t$ consistent
  with $(t^A_i , l_i)_{i \in [\sqrt{n}]}$. Hence, for any such $t$ we
  have that $R(t,1) < 2^{-\delta n}$. Thus, in this case $\bada(t) =0$
  for all such $t$ and we are done.

Now consider the other case when 
\begin{equation}\label{eq:othercase}
\prob{X\in A ~|~ X_{l_i} =1, T^A_i = t^A_i , L_i = l_i, \mbox{ for
    each } i \in [\sqrt{n}]} \geq 2^{-\delta n} .
\end{equation}

Consider a $t=(j,k_j, (t^A_i,t^B_i,l_i)_{i\in [\sqrt{n}]})$ consistent with $(t^A_i ,
  l_i)_{i \in [\sqrt{n}]}$. We know that the bit $(X_{k_j} ~|~ T=t,
  X_{l_j}=1)$ is a unbiased bit. Now, suppose $\bada(t)
=1$. Then, for some $a \in \{0,1\}$, we have  
$$\prob{ X \in A ~|~ T=t, X_{l_j} = 1, X_{k_j} = a }  <   0.99
\left(\prob{ X \in A~|~ T=t, X_{l_j} = 1 }\right).$$
By a simple rewriting of the above inequality, we have
\begin{equation}\label{eq:aislesslikely}
\prob{X_{k_j} = a ~|~ X \in A, T=t, X_{l_j} = 1} < 0.99
\left(\prob{X_{k_j} = a ~|~ T=t, X_{l_j} = 1} \right)= 0.99/2.
\end{equation}
In other words, the unbiased bit $(X_{k_j}~|~T=t, X_{l_j} = 1)$ when
conditioned on the event ``$X \in A$'' is now more likely to be $1-a$
than $a$.

\phnote{15 Jul: Added conditioning in def of E, please check if
  correct} Suppose, for contradiction, that $$\prob{\bada(T)=1 ~|~
  T^A_i = t^A_i, L_i = l_i, \text{ for each } i \in [\sqrt{n}]} \geq
\frac{1}{6400}. $$ Consider the random variable $$Z\defeq (X~|~X_{l_i}
=1, T^A_i = t^A_i , L_i = l_i, \mbox{ for each } i \in [\sqrt{n}]). $$
Note that the distribution of $Z$ is uniform and each string has
probability $\left(\frac{1}{\binom{r}{r/2}}\right)^{\sqrt{n}}$. \phnote{15 Jul: Added conditioning in
  def of E, please check if correct} Consider the event $E\defeq
``X\in A ~|~ T=t, X_{l_j}=1"$, which by \eqref{eq:othercase} has
probability at least $2^{-\delta n}$. Therefore the probability of
each string in the distribution $(Z|E)$ would be at most $2^{\delta n
} \cdot \left(\frac{1}{\binom{r}{r/2}}\right)^{\sqrt{n}}$. Therefore,
using standard estimates on binomial coefficients,
\begin{align*}
H(Z|E) & \geq   \sqrt{n} \cdot   \log {\binom{r}{r/2}}  - \delta n \geq \sqrt{n} \cdot r (1 - o(1)) - \delta n  .
\end{align*}

%By definition of the distribution $\mu$, $Z$ is the
%uniform distribution on $(r+1)\sqrt{n}$ bits. Now, further condition
%$Z$ on the event $E\defeq ``X\in A"$, which by \eqref{eq:othercase} has
%probability at least $2^{-\delta n}$. The following shows that this
%conditioning reduces the entropy by at most $\delta n$.
%\begin{equation*}
%H(Z|E) \geq H_\infty(Z|E) \geq H_\infty(Z) - \log
%\left(\frac1{\prob{E}}\right) = (r+1)\sqrt{n} - \delta n.
%\end{equation*}

\phnote{15 Jul: Added this para}
Observe that conditioned on $T^A_i = t^A_i, L_i = l_i, \text{ for each } i \in
[\sqrt{n}]$, the index $K_J$ can equally likely be any one of the $r\sqrt{n}$
indices in $\bigcup_i t^A_i$ (each resulting in a different value for $T$). Furthermore, from \eqref{eq:aislesslikely}, we
have that whenever $\bada(T)=1$ (which by assumption happens with probability at
least 1/6400), conditioning on $E$ causes $X_{K_J}$ to be a biased bit
and hence $H(X_{K_J}) \leq H(0.99/2)$. When $\bada(T)=0$, which
occurs with probability at most $1-1/6400$ by assumption, $H(X_{K_J})$
can be trivially bounded from above by 1. Using these facts, we can upper
bound the entropy of $(Z|E)$ as follows:
\begin{align*}
H(Z|E) & \leq \sum_{i} H(Z_i | E) \qquad \qquad
\qquad \text{[By subadditivity of entropy]}\\
&\leq  r\sqrt{n} \left(\frac{H(0.99/2)}{6400} +
  \left(1-\frac{1}{6400}\right)\right) .
\end{align*}
% \begin{align*}
% &&&H(X ~|~ X\in A, X_{l_i} =1, T^B_i = t^B_i , L_i = l_i, \mbox{ for
%   each } i \in [\sqrt{n}] ) \\
% & \leq &&(r+1)\sqrt{n} \cdot H(X_{k_j} ~|~ X\in A, X_{l_i} =1, T^B_i = t^B_i , L_i = l_i, \mbox{ for
%   each } i \in [\sqrt{n}] ) \\
% &&&\qquad \qquad \qquad\qquad \qquad \qquad \qquad\qquad \qquad \text{[By subadditivity of entropy]}\\
% & = &&(r+1)\sqrt{n} \left(\frac{1}{200} \cdot H\left(X_{k_j} ~|~ X\in A, X_{l_i} =1, T^B_i = t^B_i , L_i = l_i, \mbox{ for
%   each } i \in [\sqrt{n}], \bada(T)=1 \right) \right.\\
%  &&&+ \left. \left(1-\frac1{200}\right)\cdot H\left(X_{k_j} ~|~ X\in A, X_{l_i} =1, T^B_i = t^B_i , L_i = l_i, \mbox{ for
%   each } i \in [\sqrt{n}], \bada(T)=0\right)\right)\\
% & \leq && (r+1)\sqrt{n}\cdot \left(\frac{H(0.99/2)}{200} + \left(1-\frac{1}{200}\right)\right).
% \end{align*}
Combining the upper and lower bounds on $H(Z|E)$, we get 
$$\delta n \geq (1-H(0.99/2) - o(1)) \cdot \frac{r \sqrt{n}}{6400} . $$ 

 Thus, if  $\delta>0$ is small enough we get a contradiction.
\end{proof}
The following claim shows that a version of
\lref[Lemma]{lem:oneisdominated} is true when $\bad(t)=0$. The proofs
of this claim and the subsequent claim differ significantly from the
proofs of the corresponding claims in Razborov's result~\cite{Razborov1992} of
linear lower bound for set-disjointness. This is because we need to
consider several sub-events of $U_1$. Our arguments are more general
and in fact can also be used in the context of set-disjointness. 
\begin{claim}
\label{claim:both} Let $n$ be large enough. If $\bad(t)=0$, then,
\begin{align*}
\mu(U_1 \cap W~|~T=t)  &\leq \frac{16}{3(0.99)^2}\mu(U_2 \cap W~|~T=t)  + \frac{16}{(0.99)^2} \mu(U_0 \cap W~|~T=t) + 2^{-\delta n/2}.
\end{align*}
\end{claim}

\ifnum\fsttcs=0
\begin{proof}[Proof of {\lref[Claim]{claim:both}}]
Recall the definition of $\bad(t)$ from \eqref{eq:badt}. We will
consider three cases depending on the relative sizes of $\rho_l(t)$
and $\rho_k(t)$ with respect to $\mu(U_1 \cap W ~|~ T=t)$.
\begin{itemize}
\item $\rho_l(t) > 3\mu(U_1 \cap W~|~T=t)/4$:
In this case, we have $\bada(t) \vee \badb(t) = 0$. We can now bound
$\mu(U_1 \cap W ~|~ T=t)$ as follows.
\begin{align*}
 \frac{3}{4}\cdot \mu(U_1 \cap W~|~T=t) & < \rho_l(t) \leq \prob{X_{l_j} =Y_{l_j}=1, X \in A, Y \in B ~|~ T=t} \\
& =  \prob{X_{l_j} =Y_{l_j}=1~|~ T=t}  \cdot R(t,1) \cdot C(t,1) \\
& \leq \frac{\beta^2}{4} \cdot \left(\frac{R(t,1,1)}{0.99} + 2^{-\delta n} \right)
\cdot \left(\frac{C(t,1,1)}{0.99} + 2^{-\delta n} \right) \\
& \leq \frac{\beta^2}{4(0.99)^2}(R(t,1,1)  C(t,1,1))  + 2^{-\delta n} \\
& = \frac{4}{(0.99)^2}\cdot \mu(U_2 \cap W~|~T=t) + 2^{-\delta n} .
\end{align*}
\item $\rho_k(t)> 3\mu(U_1 \cap W~|~T=t)/4$:
Similar arguments as above show 
\begin{align*}
 \frac{3}{4}\cdot \mu(U_1 \cap W~|~T=t) & < \frac{4}{(0.99)^2}\cdot \mu(U_2 \cap W~|~T=t) + 2^{-\delta n} .
\end{align*}
\item $\min\{\rho_l(t),\rho_k(t)\} \geq \mu(U_1 \cap W~|~T=t)/4$:  From
  $\rho_l(t) \geq \mu(U_1 \cap W~|~T=t)/4$, we have
\begin{align*}
\frac{1}{4}\cdot \mu(U_1 \cap W~|~T=t) & \leq \rho_l(t) \leq \prob{X_{l_j} =Y_{l_j}=1, X \in A, Y \in B ~|~ T=t} \\
& = \frac{\beta^2}{4}\cdot R(t,1) C(t,1) \\
& \leq \frac{\beta^2}{4} \cdot \left(\frac{R(t,1,0)}{0.99} + 2^{-\delta n} \right)
\cdot \left(\frac{C(t,1,0)}{0.99} + 2^{-\delta n} \right) \\
& \leq \frac{\beta^2}{4(0.99)^2}(R(t,1,0)  C(t,1,0))  + 2^{-\delta n}. 
\end{align*}
Similarly from $\rho_k(t) \geq \mu(U_1 \cap W~|~T=t)/4$, we have
\begin{align*}
\frac{1}{4} \cdot \mu(U_1 \cap W~|~T=t) & \leq \prob{X_{k_j} =Y_{k_j}=1, X \in A, Y \in B ~|~ T=t} \\
& \leq \frac{\beta^2}{4(0.99)^2}(R(t,0,1)  C(t,0,1))  + 2^{-\delta n}. 
\end{align*}
Multiplying the above two inequalities we have,
\begin{align}
& \left(\frac{1}{4} \cdot \mu(U_1 \cap W~|~T=t) \right)^2 \notag\\
& \leq \left(\frac{\beta^2}{4(0.99)^2}\cdot (R(t,1,0)  C(t,1,0))  +
  2^{-\delta n}\right) \left( \frac{\beta^2}{4(0.99)^2}\cdot (R(t,0,1)  C(t,0,1))  + 2^{-\delta n}\right)\notag\\
& \leq  \frac{\beta^4}{4^2(0.99)^4} \cdot\left(R(t,1,0)  C(t,1,0) R(t,0,1)
  C(t,0,1)\right) + 2^{-\delta n} \label{eq:before}\\
& =  \frac{\beta^4}{4^2(0.99)^4} \cdot\left(R(t,1,0)  C(t,0,1) R(t,0,1)
  C(t,1,0)\right) + 2^{-\delta n} \label{eq:after}\\
& =  \frac{4^2}{(0.99)^4} \cdot \prob{(X_{l_j}, X_{k_j},Y_{l_j},Y_{k_j}) = (0,1,1,0) , X \in A, Y \in B ~|~T=t}  \cdot \notag\\
& \qquad \qquad  \prob{(X_{l_j}, X_{k_j},Y_{l_j},Y_{k_j}) = (1,0,0,1) , X \in A, Y \in B ~|~T=t}  + 2^{-\delta n} \notag\\
& \leq \frac{4^2}{(0.99)^4} \cdot \left(\mu(U_0 \cap W~|~T=t)\right)^2 + 2^{-\delta n}.\label{eq:final}
\end{align}
Observe that \eqref{eq:after} is obtained from \eqref{eq:before} by
re-ordering the terms, which in communication complexity jargon is
more commonly referred to as the
cut-and-paste-property. \eqref{eq:final} implies,
\begin{align*}
 \frac{1}{4}\cdot \mu(U_1 \cap W~|~T=t)  \leq \frac{4}{(0.99)^2} \cdot \mu(U_0 \cap W~|~T=t) + 2^{-\delta n/2}.
\end{align*}
\end{itemize}
Combining the three cases yields the claim.
\end{proof}
\fi

The following claim argues that not much probability is lost when  $\bad(T)=1$. 

\begin{claim} \label{claim:badsmallforone}  Let $n$ be large enough. Then, 
$$\expec{t \leftarrow T}{\mu(U_1 \cap W~|~ T=t) \cdot \bad(t)} \leq
\frac{1}{100}\cdot \expec{t \leftarrow T}{\mu(W \cap U_1~|~T=t)} + 2^{-\delta n +3}.$$
\end{claim}

\ifnum\fsttcs=0
\begin{proof}[Proof of {\lref[Claim]{claim:badsmallforone}}]
  For a partition $t$, define $\badaorb(t) =1$ if either
  $\bada(t)=1$ or $\badb(t)=1$.  We first show that for all
  partitions $t$, 
\begin{equation}
\mu(U_1 \cap W ~|~ T=t) \cdot \bad(t) \leq 4\left( \rho_l(t) \cdot
  \badaorb(t) + \rho_k(t) \cdot \badaorb(t')\right).\label{eq:u1badbound}
\end{equation}

As before, we consider three cases depending on the relative sizes of
$\rho_l(t)$ and $\rho_k(t)$ with respect to $\mu(U_1 \cap W ~|~ T=t)$.
\begin{itemize}
\item $\rho_l(t) > 3\mu(U_1 \cap W~|~T=t)/4$: In this case, we have
  $\bad(t) = \badaorb(t)$. Thus,\\ $\mu(U_1 \cap W ~|~ T=t) \cdot \bad(t) \leq \frac{4}{3} \cdot \rho_l(t)
\cdot \badaorb(t)$.
\item $\rho_k(t)> 3\mu(U_1 \cap W~|~T=t)/4$: In this case, we have
  $\bad(t) = \badaorb(t')$. Thus,\\ $\mu(U_1 \cap W ~|~ T=t) \cdot \bad(t) \leq \frac{4}{3} \cdot \rho_k(t)
\cdot \badaorb(t')$.
\item $\min\{\rho_l(t),\rho_k(t)\} \geq \mu(U_1 \cap W~|~T=t)/4$: 
In this case, we have $\bad(t) \leq \badaorb(t) + \badaorb(t')$. Hence, we have
\begin{eqnarray*}
\mu(U_1 \cap W ~|~ T=t) \cdot \bad(t) &\leq \mu(U_1 \cap W ~|~ T=t)
\cdot \left(\badaorb(t) + \badaorb(t')\right)\\
& \leq 4\left( \rho_l(t) \cdot \badaorb(t) + \rho_k(t)
\cdot \badaorb(t') \right).
\end{eqnarray*}
\end{itemize}
The bound in \eqref{eq:u1badbound} follows from combining the three cases.

We now argue that 
\begin{equation}
\expec{ t \leftarrow T}{\rho_l(t) \cdot \badaorb(t)} \leq \frac{1}{800} \cdot \expec{t \leftarrow T}{\mu(W \cap
  U_1~|~T=t)} + 2^{-\delta n}.\label{eq:rholbadbound}
\end{equation}
A similar bound holds for $\expec{ t \leftarrow T}{\rho_k(t) \cdot \badaorb(t')}$. Combining these two bounds with \eqref{eq:u1badbound}
yields the statement of the claim. 

We prove \eqref{eq:rholbadbound} by first showing that for each
partition $t$, we have
\begin{equation}
\rho_l(t) \cdot \badaorb(t) \leq \frac{1}{2} \cdot
\left(R(t,1,0) \cdot C(t,1)\cdot \badb(t) + R(t,1)\cdot C(t,1,0)\cdot
  \bada(t) +  2^{-\delta n} \right) .\label{eq:rholrtct}
\end{equation}
We consider various cases depending on the values of $\bada(t)$ and
$\badb(t)$.
\begin{itemize}
\item $\bada(t) = \badb(t)$: We first bound
  $\rho_l(t)$ as follows:
\begin{eqnarray*}
\rho_l(t) & = & \prob{X_{l_j} =Y_{l_j}=1, X \in A, Y \in B,
  (X,Y) \in U_1 ~|~ T=t},\\
& \leq & \prob{X_{l_j} =Y_{l_j}=1, X_{k_j}=0, X \in A, Y \in B,  ~|~
  T=t} \\
& & \quad +\prob{X_{l_j} =Y_{l_j}=1, , Y_{k_j}=0, X \in A, Y \in B  ~|~
  T=t},\\
& = & \frac{\beta^2}{8}\left( R(t,1,0)\cdot C(t,1) + R(t,1)\cdot
  C(t,1,0)\right).
\end{eqnarray*}
\eqref{eq:rholrtct} then follows by observing that in this case
$\badaorb(t) = \bada(t) = \badb(t)$.
\item $\bada(t) = 1, \badb(t)=0$: Since $\badb(t) =0$, we have that
  $C(t,1) \leq C(t,1,0)/0.99 + 2^{-\delta n}$. We now bound $\rho_l(t)$ as follows.
\begin{eqnarray*}
\rho_l(t) & \leq & \prob{X_{l_j} =Y_{l_j}=1, X \in A, Y \in B, ~|~
  T=t},\\
& = &\frac{\beta^2}{4} \cdot R(t,1) \cdot C(t,1) \leq \frac{\beta^2}{4(0.99)} \cdot \left( R(t,1)\cdot C(t,1,0) + 2^{-\delta n} \right)
\end{eqnarray*}
\eqref{eq:rholrtct} then follows by observing that in this case
$\badaorb(t) = \bada(t)$.
\item $\bada(t) = 0, \badb(t) =1$:
This case is similar to the above case.
\end{itemize}

We now bound $\expec{ t \leftarrow T}{R(t,1,0)\cdot C(t,1)\cdot
  \badb(t)}$. We will bound this expectation by setting the random
variable $T$ in stages: we first set $t_B=\{t^B_i , l_i  : ~ i \in
[\sqrt{n}]\}$, and then set the variable $k_j \in [n]$ from the distribution $(K_J ~|~ T_B = t_B)$.  We observe that $C(t,1)$ is only a function of $t_B$ and independent of $k_j$; thus, $C(t,1) = c(t_B)$ for some function
$c$. Similarly $R(t,1,0)$ is only a function of $t_B$ and is independent of $k_j$; thus, $R(t,1,0) = r(t_B)$ for some function $r$. We have $\badb(t) = b(t_B,k_j)$ for some function $b$. In this notation,
\lref[Claim]{claim:badsmall} states that for all $t_B$, $\expec{ k_j
  \gets K_J|T_B=t_B}{b(t_B,k_j)} \leq 1/6400$. 
%Now,
%\begin{eqnarray*}
%\expec{ t \gets T}{R(t,1,0)\cdot C(t,1)\cdot \badb(t)}
%& = & \expec{t_B \gets T_B}{c(t_B)\cdot \expec{(k_j,d_j)\gets (K_J, D_J)| T_B =
%    t_B}{ b(t_B,k_j) \cdot R(t,1,0)}}\\
%& \leq  2 \cdot & \expec{t_B \gets T_B}{c(t_B) \cdot r(t_B) \cdot \expec{k_j\gets K_J| T_B =
%    t_B}{ b(t_B,k_j)}}\\ 
%& \leq & \frac{1}{6400} \cdot  \expec{t_B \gets T_B}{c(t_B)\cdot r(t_B)}\\
%& =& \frac{1}{6400} \expec{t\gets T}{R(t,1,0) \cdot C(t,1)}\\
%& \leq & \frac{8}{6400} \expec{ t\gets T}{\mu(U_1 \cap W | T=t)}.
%\end{eqnarray*}
\begin{eqnarray*}
\expec{ t \gets T}{R(t,1,0)\cdot C(t,1)\cdot \badb(t)}
& = & \expec{t_B \gets T_B}{c(t_B)\cdot r(t_B) \cdot \expec{k_j\gets K_J| T_B =
    t_B}{ b(t_B,k_j) }}\\
& = & \expec{t_B \gets T_B}{c(t_B) \cdot r(t_B) \cdot \expec{k_j\gets K_J| T_B =
    t_B}{ b(t_B,k_j)}}\\ 
& \leq & \frac{1}{6400} \cdot  \expec{t_B \gets T_B}{c(t_B)\cdot r(t_B)}\\
& =& \frac{1}{6400} \cdot \expec{t\gets T}{R(t,1,0) \cdot C(t,1)}\\
& \leq & \frac{8}{6400} \cdot \expec{ t\gets T}{\mu(U_1 \cap W | T=t)}.
\end{eqnarray*}

% We now bound $\expec{ t \leftarrow T}{R(t,1,0)\cdot C(t,1)\cdot
%   \badb(t)}$.

% %\phnote{I am not fully convinced that $R(t,1,0)$ is fixed below.}

% Fix $\{t^B_i , l_i  : ~ i \in [\sqrt{n}]\}$. Below we condition on the event $T^B_i = t^B_i , L_i = l_i, \mbox{ for each } i \in [\sqrt{n}]$. 
%  Note that under this conditioning the values of $C(t,1)$ are the same (say $c$) for every $t$ consistent with $\{t^B_i , l_i  : ~ i \in [\sqrt{n}]\}$. Also for every $k_j$, $\expec{d_j \leftarrow D_J|(K_J = k_j)}{R(t,1,0)}$ is the same (say $r$) . Therefore,  
% \begin{align*}
% & \expec{t \leftarrow T}{R(t,1,0) \cdot C(t,1) \cdot \badb(t)} \\
% & = \expec{k_j \leftarrow K_J}{\left(\expec{d_j \leftarrow D_J|(K_J = k_j)}{R(t,1,0)}\right) \cdot \left(\expec{d_j \leftarrow D_J|(K_J = k_j)}{ C(t,1) \cdot \badb(t)}\right)} \\
% & =  (c \cdot r) \expec{t \leftarrow T}{\badb(t)} \\
% & \leq   \frac{(c \cdot r)}{6400} \quad \mbox{(from \lref[Claim]{claim:badsmall})}\\
% & =  \frac1{6400}\cdot  \expec{t \leftarrow T}{R(t,1,0) \cdot C(t,1)}\\
% & \leq \frac{8}{6400} \cdot \expec{t \leftarrow T}{\mu(U_1 \cap W ~|~ T=t) }
% \end{align*}
Hence, $$\expec{t \leftarrow T}{R(t,1,0) \cdot C(t,1) \cdot \badb(t)}
\leq \frac1{800}\expec{t \leftarrow T}{\mu(U_1 \cap W ~|~ T=t )}.$$ A similar
bound holds for $\expec{t \leftarrow T}{R(t,1) \cdot C(t,1,0) \cdot
  \bada(t)}$. Combining these bounds with \eqref{eq:rholrtct} yields
\eqref{eq:rholbadbound} which completes the proof of the claim.

\end{proof}
\fi

\ifnum\fsttcs=1
For want of space, the 
proofs of Claims~\ref{claim:both}-\ref{claim:badsmallforone} are
deferred to the full version of the paper~\cite{HarshaJ2013-arXiv}.
\fi

\lref[Lemma]{lem:oneisdominated} follows by combining
\lref[Claim]{claim:both} and \lref[Claim]{claim:badsmallforone} as
follows.

\begin{align*}
& 0.99 \mu(U_1 \cap W)\\
 = & 0.99 \expec{t \gets T}{\mu(U_1 \cap W ~|~
  T=t)}\\
\leq & \expec{t \gets T}{\mu(U_1 \cap W ~|~ T=t) \cdot (1-\bad(t))} + 2^{-\delta n+3}
\qquad \text{(from \lref[Claim]{claim:badsmallforone})}\\
 \leq & \expec{t\gets T}{\left(\frac{16\mu(U_2 \cap
    W~|~T=t)}{3(0.99)^2} + \frac{16 \mu(U_0 \cap W~|~T=t)}{(0.99)^2}  +
    2^{-\delta n/2}\right) (1-\bad(t))} \\
&\qquad \qquad \quad \quad+ 2^{-\delta n+3} \qquad \text{(from \lref[Claim]{claim:both})}\\
 \leq & \frac{16}{3(0.99)^2}\cdot \mu(U_2 \cap
    W)  + \frac{16}{(0.99)^2} \cdot \mu(U_0 \cap W) +
    2^{-\delta n/2 +1}
\end{align*}
\qed

\section*{Acknowledgements}

We thank Jaikumar Radhakrishnan for several useful discussions and the
anonymous reviewers for useful comments.

{\small
\ifnum\fsttcs=1

\else
%\bibliographystyle{prahladhurl}
%\bibliography{../tribes-bib}
\newcommand{\etalchar}[1]{$^{#1}$}

\fi
}

\end{document}